\documentclass[12pt]{article}

\usepackage[letterpaper, total={6in, 8.5in}]{geometry}
\usepackage{hyperref}

\hypersetup{
  hidelinks,
  colorlinks=true,
  allcolors=black,
  pdfstartview=Fit,
  breaklinks=true
}

\usepackage{epsfig}
\usepackage{amsmath}
\usepackage{amssymb}
\usepackage{color}
\usepackage{url}
\newlength{\figurewidth}
\newlength{\smallfigurewidth}

\author{%
Magnus Berg$^{\ast}$\hspace{3mm}  Shahin Kamali$^{\dag}$ \hspace{3mm} Katherine Ling$^{\dag}$ \hspace{3mm}  Cooper Sigrist$^{\ddag}$\\[1em]
{\small
\begin{minipage}{\linewidth}\begin{center}
\begin{minipage}{.33\linewidth}
\begin{center}
$^{\ast}$University of Southern Denmark \\
Campusvej 55 \\
5230 Odense M, Denmark \\
\url{magbp@imada.sdu.dk} 
\end{center}
\end{minipage}
\begin{minipage}{.32\linewidth}
\begin{center}
$^{\dag}$York University \\
Toronto, Ontario \\
Canada, M3J1P3 \\
\url{kamalis@yorku.ca} \\
\url{kathling@my.yorku.ca}
\end{center}
\end{minipage}
\begin{minipage}{.32\linewidth}
\begin{center}
$^{\ddag}$University of Massachusetts Amherst\\
1 Campus Center Way \\
Amherst, MA, 01003, USA \\ \url{csigrist@umass.edu} 
\end{center}
\end{minipage}
\vspace{.5cm} \\ 
\end{center}
\end{minipage}}
}
\date{}

\usepackage{enumerate}
\usepackage{enumitem}
\usepackage{tikz}

\usepackage{algpseudocode}
\usepackage{algorithm}
\algtext*{EndIf}
\algtext*{EndWhile}
\algtext*{EndFunction}
\algrenewcommand\algorithmicfunction{\textbf{query}}

\newtheorem{theorem}{Theorem}
\newtheorem{corollary}{Corollary}

\newtheorem{proposition}{Proposition}
\newtheorem{lemma}{Lemma}
\newenvironment{proof}{\begin{trivlist}\item[]{\it Proof }}%
{\hspace*{\fill}\raisebox{-1pt}{\boldmath$\Box$}\end{trivlist}}

\usepackage{subcaption}

\usepackage{bm}
\usepackage{xcolor}
\usepackage{mathdots}

\newcommand{\on}[1]{\ensuremath{\operatorname{#1}}}

\begin{document}
\input{figs}
\newcommand{\algoBelowSlicing}{\begin{center}
\scalebox{.9}{
\begin{minipage}{13.5cm}
\begin{algorithm}[H]
\caption{The \textbf{below} Query for General Polyominoes}\label{alg:slicing-below}
\begin{algorithmic}

\Function{$\bm{{\on{is\_bot}}}_P$}{$c$}
\If{$\bm{{\on{depth_{T^*}}}}(c) = i^*$}
    \If{$\bm{{\on{level\_rank}}}_{T^*}(c) \leq \on{FirstBot}$}
        \State \Return{\texttt{true}}
    \EndIf
\EndIf
\If{$\bm{{\on{depth_{T^*}}}}(c) = f(n)$}
    \State \Return{\texttt{true}}
\EndIf
\State \Return{\texttt{false}}
\EndFunction
\end{algorithmic}
\hrulefill
\begin{algorithmic}
\Function{$\bm{{\on{below}}}_P$}{$c$}
\If{$\bm{{\on{is\_bot}}}_P(c) = \texttt{false}$}
    \If{$\bm{{\on{parent}}}_{T^*}(c+1) = c$}
        \State \Return{$\bm{{\on{child}}}_{T^*}(c,1)$}
    \Else 
        \State \Return{\texttt{null}}
    \EndIf
\EndIf
\If{$\bm{{\on{depth}}}_{T^*}(c) = i^*$}
    \State $k \leftarrow \bm{{\on{level\_rank}}}_{T^*}(c) -1$
\EndIf
\If{$\bm{{\on{depth}}}_{T^*}(c) = f(n)$}
    \State $k \leftarrow \bm{{\on{level\_rank}}}_{T^*}(c) + \on{FirstBot} -1$
\EndIf
\If{$\on{Bot}[k] = 0$}
    \State \Return{\texttt{null}}
\EndIf
\State $r_{c'} \leftarrow \bm{{\on{select}}}_{\on{Top}}(k)$
\State \Return{$\bm{{\on{level\_select}}}_{T^*}(1,\on{FirstTop}+r_{c'}+1)$}
\EndFunction
\end{algorithmic}
\end{algorithm}
\end{minipage}}
\end{center}

}

\newcommand{\algoAboveSlicing}{\begin{center}
\scalebox{.9}{
\begin{minipage}{13.5cm}
\begin{algorithm}[H]
\caption{The \textbf{above} Query for General Polyominoes}\label{alg:slicing-above}
\begin{algorithmic}
\Function{$\bm{{\on{is\_top}}}_P$}{$c$}
\If{$\bm{{\on{depth_{T^*}}}}(c) = 1$}
    \If{$\bm{{\on{level\_rank}}}_{T^*}(c) > \on{FirstTop}$}
        \State \Return{\texttt{true}}
    \EndIf
\Else 
    \State \Return{\texttt{false}}
\EndIf
\EndFunction
\end{algorithmic}
\hrulefill 
\begin{algorithmic}
\Function{$\bm{{\on{above}}}_P$}{$c$}
\If{$\bm{{\on{is\_top}}}_P(c) = \texttt{false}$}
    \If{$\bm{{\on{child}}}_{T^*}(\bm{{\on{parent}}}_{T^*}(c),1) = c$}
        \State \Return{$\bm{{\on{parent}}}_{T^*}(c)$}
    \Else 
        \State \Return{\texttt{null}}
    \EndIf
\EndIf
\State $k \leftarrow \bm{{\on{level\_rank}_{T^*}}}(c) - \on{FirstTop} -1$
\If{$\on{Top}[k] = 0$}
    \State \Return{\texttt{null}}
\EndIf
\State $r_{c'} \leftarrow \bm{{\on{select_{Bot}}}}(k)$
\If{$r_{c'} < \on{FirstBot}$}
    \State \Return{$\bm{{\on{level\_select}}}_{T^*}(i^*, r_{c'} + 1)$}
\Else 
    \State \Return{$\bm{{\on{level\_select}}}_{T^*}(f(i), r_{c'} - \on{FirstBot} + 1)$}
\EndIf
\EndFunction
\end{algorithmic} 
\end{algorithm}
\end{minipage}}
\end{center}}

\newcommand{\algoVisSlicing}{
\begin{center}
\scalebox{0.9}{
\begin{minipage}{13.5cm}
\begin{algorithm}[H]
\caption{The \textbf{is\_visible} Query for general polyominos}
\label{alg:is_vis_slicing}
\begin{algorithmic}
\Function{$\bm{{\on{is\_visible}}}_{P}$}{$c_1,c_2$}
\If{$\bm{{\on{is\_visible}}}_{p_d}(c_1,c_2) = \texttt{true}$}
    \State \Return{\texttt{true}}
\EndIf
\State{$c' \leftarrow c_2$}
\While{$c' \geq c_1$}
\If{$c'-c_1 = \bm{{\on{depth}}}_{T^*}(c') - \bm{{\on{depth}}}_{T^*}(c_1)$} \Comment{Step (i)}
    \State \Return{\texttt{true}}
\Else
    \State $anc \leftarrow \bm{{\on{ancestor}}}_{T^*}(c',1)$ 
    \If{$c'-anc \neq  \bm{{\on{depth}}}_{T^*}(c') - \bm{{\on{depth}}}_{T^*}(anc)$} \Comment{Step (ii)}
        \State \Return{\texttt{false}}
    \Else \Comment{Step (iii)}
        \State $c' \leftarrow \bm{{\on{above}}}_{P}(anc)$
        \If{$c' = \texttt{null}$} 
            \State \Return{\texttt{false}}
        \EndIf
    \EndIf
\EndIf
\EndWhile
\State \Return{\texttt{false}}
\EndFunction
\end{algorithmic}
\end{algorithm}
\end{minipage}
}
\end{center}
}

\newcommand{\algoNeighbourhoodBar}{
\begin{center}
\scalebox{.9}{
\begin{minipage}{12cm}
\begin{algorithm}[H]
\caption{Neighbourhood Queries for Bar Graphs}
\begin{algorithmic}
\Function{$\bm{{\on{above}}}_G$}{$x$}
    \If{$S_G[x+1] = 1$}
        \State \Return{\texttt{null}}
    \Else 
        \State \Return{$x+1$}
    \EndIf
\EndFunction
\end{algorithmic}
\hrulefill 
\begin{algorithmic}
\Function{$\bm{{\on{below}}}_G$}{$x$}
    \If{$S_G[x] = 1$}
        \State \Return{\texttt{null}}
    \Else 
        \State \Return{$x-1$}
    \EndIf
\EndFunction
\end{algorithmic}
\hrulefill 
\begin{algorithmic}
\Function{$\bm{{\on{left}}}_G$}{$x$}
    \If{$\bm{{\on{bar}}}_G(x) = 0$}
        \State \Return{\texttt{null}}
    \EndIf
    \State $c_\ell \leftarrow \bm{{\on{select}}}_{S_G}(\bm{{\on{bar}}}_G(x)-1)$
    \If{$c_\ell + \bm{{\on{height}}}_G(x) < \bm{{\on{first\_in\_bar}}}_G(x)$}
        \State \Return{$c_\ell+\bm{{\on{height}}}_G(x)$}
    \Else 
        \State \Return{$\texttt{null}$}
    \EndIf
\EndFunction
\end{algorithmic}
\hrulefill 
\begin{algorithmic}
\Function{$\bm{{\on{right}}}_G$}{$x$}
\State $c_r \leftarrow \bm{{\on{select}}}_{S_G}(\bm{{\on{bar}}}_G(x)+1)$
\If{$c_r = \texttt{null}$}
    \State \Return{\texttt{null}}
\EndIf
\If{$c_r + \bm{{\on{height}}}_G(x) \leq \bm{{\on{size}}}_G(\bm{{\on{bar}}}_G(x)+1)$}
    \State \Return{$c_r + \bm{{\on{height}}}_G(x)$}
\Else 
    \State \Return{\texttt{null}}
\EndIf
\EndFunction
\end{algorithmic}
\end{algorithm}
\end{minipage}
}
\end{center}
}

\newcommand{\algoVisBar}{
\scalebox{.9}{
\begin{minipage}{18cm}
\begin{algorithm}[H]
\caption{Visibility Query for Bar Graphs}
\label{alg:is_vis_bar}
\begin{algorithmic}
\Function{$\bm{{\on{is\_visible\_block}}}_{G}$}{$x_1,x_2$}
    \State $q \leftarrow \bm{{\on{select}}}_{B_G}(i)$
    \State $entry \leftarrow S_G[q, q + k]$
    \State $y \leftarrow x_2$
    \If{$q + k \geq x_2$} \Comment{$x_2$ is a leftover cell}
        \State $y \leftarrow \bm{{\on{left}}}_G(x_2)$
    \EndIf
    \If{$y = \texttt{null}$}
        \State \Return{$\texttt{false}$}
    \Else 
        \State \Return{$R_k[entry].\bm{{\on{vis}}}(x-q,y-q)$}
    \EndIf
\EndFunction
\end{algorithmic}
\hrulefill 
\begin{algorithmic}
\Function{$\bm{{\on{is\_visible}}}_{G}$}{$x_1,x_2$}
\If{$\bm{{\on{bar}}}_G(x_1) = \bm{{\on{bar}}}_G(x_2)$} \Comment{$x_1$ and $x_2$ belongs to the same bar}
\State \Return{\texttt{true}}
\EndIf
\If{$\bm{{\on{height}}}_G(x_1) \neq \bm{{\on{height}}}_G(x_2)$} \Comment{$x_1$ and $x_2$ have different height}
\State \Return{\texttt{false}}
\EndIf
\If{$S_G[x_1] = S_G[x_2] = 1$} \Comment{$x_1$ and $x_2$ are the bottommost cells in the bar}
\State \Return{\texttt{true}}
\EndIf
\State $i \leftarrow \bm{{\on{block}}}_G(x_1)$
    \If{$\bm{{\on{block}}}_G(x_2) = i$} \Comment{$x_1$ and $x_2$ belong to the same block}
        \State \Return{$\bm{{\on{is\_visible\_block}}}_{G}(x_1, x_2)$}
    \Else \Comment{$x_1$ and $x_2$ belong to different blocks}
        \State $j \leftarrow \bm{{\on{block}}}_G(x_2)$
        \State $h \leftarrow \bm{{\on{height}}}_G(x_1)$
        \State $c_{last_i} \leftarrow \bm{{\on{select}}}_{B_G}(i+1)-1$
        \State $w \leftarrow \texttt{null}$
        \If{$\bm{{\on{height}}}_G(c_{last_i}) \geq h$}
            \State $w \leftarrow \bm{{\on{first\_in\_bar}}}_G(c_{last_i}) + h -1$
        \EndIf
        \State{$c_{first_j} \leftarrow \bm{{\on{select}}}_{B_G}(j)$}
        \State $z \leftarrow \texttt{null}$
        \If{$\bm{{\on{height}}}_G(c_{first_j}) \geq h$}
            \State $z \leftarrow \bm{{\on{first\_in\_bar}}}_G(c_{first_j}) + h - 1$
        \EndIf
        \If{($w = \texttt{null}) \textbf{ or } (z = \texttt{null})$}
            \State \Return{\texttt{false}}
        \EndIf
        \If{$(\bm{{\on{is\_visible\_block}}}_{G}(x_1, w) = \texttt{false}) \textbf{ or } (\bm{{\on{is\_visible\_block}}}_{G}(z, x_2) = \texttt{false})$}
            \State \Return{\texttt{false}}
        \EndIf
        \State $\ell \leftarrow \bm{{\on{LCA}}}_{C_G}(i, j)$
        \State $h_1 \leftarrow R_k[B_\ell].height$ 
        \State $c_{last_\ell} \leftarrow \bm{{\on{select}}}_{b_G}(\ell+1)-1$
        \State $h_2 \leftarrow \bm{{\on{height}}}(c_{last_\ell})$ 
        \State $lo \leftarrow \on{min}\{h_1, h_2\}$ \Comment{height of smallest bar in block $B_\ell$}
        \If{$lo \geq h$}
            \State \Return{\texttt{true}}
        \Else 
            \State \Return{\texttt{false}}
        \EndIf
    \EndIf
\EndFunction
\end{algorithmic}
\end{algorithm}
\end{minipage}
}
}

\pagenumbering{arabic}
\title
{\Large
\textbf{Space-Efficient Data Structures for \\ Polyominoes and Bar Graphs \thanks{A one-page summary of a preliminary result in this paper was announced  in~\cite{DCC21}.}}
}

\maketitle

\begin{abstract}
We provide a compact data structure for representing polyominoes that supports neighborhood and visibility queries. Neighborhood queries concern reporting adjacent cells to a given cell, and visibility queries determine whether a straight line can be drawn within the polyomino that connects two specified cells.
For an arbitrary small $\epsilon >0$, our data structure can encode a polyomino with $n$ cells in $(3+\epsilon)n + o(n)$ bits while supporting all queries in constant time. The space complexity can be improved to $3n+o(n)$, while supporting neighborhood queries in $\mathcal{O}(1)$ and visibility queries in $\mathcal{O}(t(n))$ for any arbitrary $t(n) \in \omega(1)$. 
Previous attempts at enumerating polyominoes have indicated that at least $2.00091n - o(n)$ bits are required to differentiate between distinct polyominoes, which shows our data structure is compact.

In addition, we introduce a succinct data structure tailored for bar graphs, a specific subclass of polyominoes resembling histograms. We demonstrate that a bar graph comprising $n$ cells can be encoded using only $n + o(n)$ bits, enabling constant-time query processing. Meanwhile, $n-1$ bits are necessary to represent any bar graph, proving our data structure is succinct.
\end{abstract}

\section{Introduction}

A \emph{polyomino} is a mathematical construct consisting of unit \emph{cells} that adhere to one another along their edges, with their corners aligning.
We do not impose any constraints on the shape of polyominoes. In particular, polyominoes may be disconnected or encompass \emph{holes}, i.e.\ the complement of a polyomino can be disjoint
(see Figure~\ref{fig:polyomino_intro}).

Polyominoes find practical applications in modelling planar environments with obstacles~\cite{GarciaYKRB22,Anh2018}. 
In this context, each cell of a polyomino represents a navigable location, while missing cells denote walls or obstacles. When managing a fleet of one or more robots operating within such an environment, it is often essential to perform fundamental tasks such as finding the shortest path from a robot to a target cell, constructing a minimum spanning tree encompassing all robots, and determining whether a robot can \emph{see} a specific point within the polyomino, as applied to the art gallery problem~\cite{GarciaYKRB22,AbrahamsenAM22}. To accomplish these tasks, basic algorithms, like Dijkstra's shortest path algorithm, rely on \emph{navigation}
queries performed on the \emph{dual graph} of the polyomino. This dual graph is an unweighted graph that assigns a vertex to each cell in the polyomino, with edges connecting adjacent cells.
Navigation queries encompass:
\begin{itemize}
    \item \emph{neighborhood} queries to identify adjacent cells to a given cell $c$ in the left, right, bottom, and top directions, if they exist.
    \item \emph{adjacency} queries to determine if two cells, $c$ and $c'$, are connected.
    \item \emph{degree} queries to report the number of adjacent cells to a specific cell $c$.
\end{itemize}
Note that an efficient solution for answering neighborhood queries can also handle adjacency and degree queries. 
This follows as each cell in a polyomino has a degree at most 4.

\begin{figure}
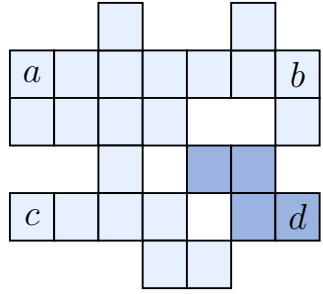

\centering
\scalebox{.95}{\shahinFigOneSingleColour} 
\caption{A polyomino with $n=25$ cells: Cell $a$ is visible to $b$ but not to $c$ or $d$. Note that the dark cells form a disconnected component of the polyomino.
\vspace*{-0.2cm}}
\label{fig:polyomino_intro}
\end{figure}

\subsection{Contribution}
We present space-efficient data structures for polyominoes that support neighborhood (and thus other navigation) queries and visibility queries.
Basic dual graph representations like the adjacency matrix and adjacency list are suboptimal due to their lack of compactness and/or inefficiency in query response.
A data structure is said to be \emph{compact} if it takes $\mathcal{O}(B)$ to encode a member
from a family of objects, e.g., the family of all polyominoes of size $n$, while $\Omega(B)$ bits are necessary for any such encoding.
Similarly, a structure is considered \emph{succinct} if it requires $B +o(B)$ bits to encode a member of the family, with $B - o(B)$ bits being necessary.

Our contributions can be summarized as follows:

\begin{itemize}
\item In a special case where the polyomino can fit within a strip of width or height $o(n)$, our data structure stores a polyomino of size $n$ using $3n+o(n)$ bits while supporting queries in $\mathcal{O}(1)$ (Theorem~\ref{th:mainBounded}).
\item In the general setting, for any $\epsilon >0$, our data structure stores a polyomino of size $n$ using $(3+\epsilon)n + o(n)$ bits while supporting queries in constant time (Theorem~\ref{th:mainGeneral} and Corollary~\ref{corolOne}).
Moreover, we can improve the space complexity to $3n+o(n)$ while supporting queries in $\mathcal{O}(t(n))$ time for any arbitrarily small $t(n) \in \omega(1)$ (Theorem~\ref{th:mainGeneral} and Corollary~\ref{corolTwo}).
. Prior enumeration efforts have established that distinguishing between polyominoes requires a minimum of $2.00091n-o(n)$ bits 
(likely $2.021 n -o(n)$)~\cite{JG00,BarequetRS16}. 
This corroborates the compactness of our data structures.

\item Additionally, we investigate a class of polyominoes known as \emph{bar graphs}~\cite{BLPS96,Fer96}, which are a fundamental subset of polyominoes that resemble histograms and require $n-1$ bits to encode. We introduce a succinct data structure for bar graphs that require $n+o(n)$ bits of storage while supporting constant-time neighborhood and visibility queries (Theorem~\ref{th:barGraph}).

\end{itemize}

Throughout the paper,
we assume word RAM model with a word size of $\Omega(\log n)$. 

\vspace*{1mm}
\section{Preliminaries}
\vspace*{-1mm}

Our data structure for polyominoes has two components: an ordinal tree and a bitstring. Succinct representations for both already exist, and we briefly review the main results related to these. 

Given a bitstring $S$, the query $\bm{{\on{rank}}}_S(i)$ reports the number of occurrences of $1$, 
up to index $i$, in $S$, and $\bm{{\on{select}}}_S(i)$ reports the index of the $i$'th occurrence of $1$ in $S$.
There are succinct data structures that store the sequence in $n + o(n)$ bits and supports rank/select in constant time \cite{Jacobson89,RamanRS07,BarbayGNN10} (see also~\cite{Munro96,NavarroBook}):

\begin{lemma}\label{thm:succinct_bit_string}\cite{BarbayGNN10}
Given a bitstring $S$ of length $n$ bits, it is possible to support queries $\bm{{\on{rank}}}_S(i)$ and $\bm{{\on{select}}}_S(i)$ in $\mathcal{O}(1)$ for any $i \leq n$, in $n+o(n)$ bits of space.
\end{lemma}

\begin{lemma}\label{thm:succinct_bit_string_k}\cite{RamanRS07}
Given a bitstring $S$ of length $n$ bits with $k$ bit of value 1, it is possible to support queries $\bm{{\on{rank}}}_S(i)$ and $\bm{{\on{select}}}_S(i)$ in $\mathcal{O}(1)$ for any $i \leq n $, in $\log \binom{n}{k} + o(n)$ bits of space.
\end{lemma}


An \emph{ordinal tree} $T$ consists of a root and a (possibly empty) ordered set of children, which are roots of ordinal subtrees themselves. 
We reference each node in an ordinal tree by its index in the pre-order traversal of the tree.
Given an ordinal tree $T$, any vertices $u$ and $v $ of $T$, any $d \in\{0,1,\ldots,\on{depth}(T)\}$, and any $i \in \{1,\ldots,n\}$, consider the following queries: 

\begin{center}
\scalebox{.85}{
    \begin{tabular}{|l|l|}
    \hline 
        $\bm{{\on{parent}}}_T(v)$, $\bm{{\on{child}}}_T(v,i)$ &   the parent of $v$ and  the $i$'th child of $v$ in $T$, respectively \\ \hline  
        $\bm{{\on{depth}}}_T(v)$ & the depth of $v$ in $T$ (no.\ edges on the path from the root to $v$)  \\ \hline 
        $\bm{{\on{ancestor}}}_T(v,i)$ & the ancestor of $v$ in $T$ at level  $\bm{{\on{depth}}}_T(v)-i$ \\ \hline 
        $\bm{{\on{level\_pred}}}_T(v)$/$\bm{{\on{level\_succ}}}_T(v)$ & the predecessor/successor of $v$ in the level-order traversal of $T$   \\ \hline 
        $\bm{{\on{level\_order\_rank}}}_T(v)$ & the index of $v$ in the level-order traversal of $T$  \\ \hline 
        $\bm{{\on{level\_order\_select}}}_T(i)$ & the $i$'th node in the level-order traversal of $T$ \\ \hline 
        $\bm{{\on{level\_rank}}}_T(v)$ &  the number of nodes at the same level but before $v$  in $T$ \\ \hline $\bm{{\on{level\_select}}}_T(\ell,i)$ & the $i$'th node in level $\ell$ of $T$  \\ \hline
        $\bm{{\on{LCA}}}_T(v, u)$ & the least common ancestor of $v$ and $u$  \\ \hline
    \end{tabular}
} 
\end{center}

There are multiple succinct data structures (see~\cite{DelprattRR06,Raman013,HeMS12,NavarroS14}) for ordinal trees. We use the data structure from~\cite{He2020}, which supports translations between level-order traversals and pre-order traversals of the tree in constant time.

\begin{lemma}
\label{thm:succinct_ordinal_tree} ~\cite{He2020}	It is possible to store an ordinal tree of size $n$ in $2n + o(n)$ bits and support all queries listed above in $\mathcal{O}(1)$.
\end{lemma}

\section{Compact Data Structure for Polyominoes}
Before presenting our main result for general polyominoes, we present a data structure for \emph{nice polyominoes}, which are polyominoes that can fit entirely within a strip of height $h(P) = o(n)$. 
The same result can be achieved for polyominoes fitting within a strip of width $o(n)$ by rotating them by 90 degrees.

Given a nice polyomino $P$, 
the \emph{augmented polyomino} of $P$, denoted $P^\ast$, is the polyomino of size $n + h(P) + 1$ obtained by appending a column of size $h(P)$ comprised of \emph{dummy nodes} to the left of $P$ so that the leftmost cell of $P^\ast$ in every row of the underlying grid is a dummy node. 
No cell in $P$ is adjacent to a dummy node.
We also add a dummy node on top of the topmost dummy node and call it the \emph{root} of $P^\ast$ (see Figure~\ref{fig:augmented_polyomino}). 
Define the \emph{level} of each cell as its vertical distance to the root. 
We refer to two vertically-connected cells as a \emph{domino} with the top and bottom cells being the \emph{head} and \emph{tail} of the domino, respectively. 

Our data structure is composed of the following components: 
\begin{itemize}
    \item[] \textbf{Covering Tree $\bm{{T^\ast}}$}: From the augmented polyomino $P^\ast$, we form an ordinal tree $T^*$ called the \emph{covering tree} of $P$ as follows. The root of $T^\ast$ is the root of $P^\ast$, and the parent of each cell $c$ at level $\ell \geqslant 1$ is the rightmost cell on level $\ell-1$ that appears in the same column or on the left of $c$ and is the head of a domino. The presence of dummy nodes ensures that such a parent always exists. We store $T^\ast$ in $2 (n + h(P)+1) + o(n) = 2n +o(n)$ bits of space using the succinct data structure of Lemma~\ref{thm:succinct_ordinal_tree}. 
\item[] \noindent\textbf{Left Bitstring $\bm{{L}}$}: The \emph{left bitstring} $L$ of $P$ is a bitstring of length $n + h(P) + 1$, which, for all cells $c \in P^\ast$ stores a bit indicating whether $c$ is adjacent to a cell on its left. The bits appear in the level-order traversal of $T^\ast$. The bitstring $L$ is stored using $n + o(n)$ bits of space, using the succinct data structure of Lemma~\ref{thm:succinct_bit_string}. 

\end{itemize}

In our data structure, cells are represented by their index in the pre-order traversal of the covering tree. 
The fact that the data structure from~\cite{He2020} supports translations between pre-order and level-order traversals in constant time implies that we can efficiently identify the bit in the left bitstring corresponding to any given cell.

\begin{figure}[t]
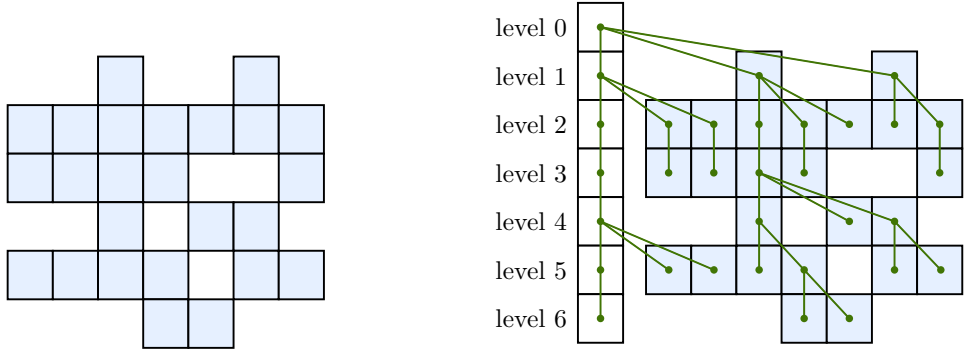

\centering
\scalebox{.97}{\ppAndAugmentedPpColoured} 
\caption{ (Left) A polyomino with $n=25$ cells. 
(Right) The covering tree of the augmented polyomino. The portion of the left bitstring $L$ associated with level $5$ 
is $L = 0011101$.
\vspace*{-0.2cm}}
\label{fig:augmented_polyomino} 
\end{figure} 

\begin{theorem}\label{th:mainBounded}
There exists a data structure that stores a polyomino $P$ that fits entirely within a strip of height or width $o(n)$, in $3n + o(n)$ bits of space while supporting navigation and visibility queries in constant time.
\end{theorem}

\begin{proof}
Summing the size $2n+o(n)$ of the covering tree and $n+o(n)$ of the left bitstring, the size of our data structure is $3n +o(n)$. Next, we show how queries are answered using operations that are supported in $\mathcal{O}(1)$ by
Lemmas~\ref{thm:succinct_bit_string} and \ref{thm:succinct_ordinal_tree}. 
In what follows, given a cell $c$ (represented by its pre-order index in $T^*$), we let $w_c \leftarrow  \bm{{\on{level\_order\_rank}}}_{T^*}(c)$. 
\begin{itemize}
\item[] {\textbf{Navigation Queries:}} 
\noindent The left neighbor of $c$, if it exists, is its level-predecessor, in which case $L[w_c] = 1$. 
Similarly, the right neighbor of $c$, if it exists, is its level-successor, in which case $L[w_c+1] = 1$. 
Moreover, by construction of $T^*$, if there exists a cell $c'$ above $c$, then $c$ must be the first child of $c'$.  
Finally, the below neighbour of $c$ is the leftmost child of $c$, which exists if the next node in the pre-order traversal is a child of $c$. 

\begin{align*}
    \bm{{\on{left}}}_P(c) &= 
    \begin{cases}
        \bm{{\on{level\_pred}}}_{T^*}(c) ,&\mbox{if $L[w_c] = 1$,} \\
        \texttt{null}, & \mbox{otherwise.}   \\
    \end{cases} \ \\ \ \\ 
    \bm{{\on{right}}}_P(c) &= 
    \begin{cases}
        \bm{{\on{level\_succ}}}_{T^*}(c) ,&\mbox{if $L[w_c+1] = 1$,} \\
        \texttt{null}, & \mbox{otherwise.} \\
    \end{cases}\ \\ \ \\ 
    \bm{{\on{above}}}_P(c) &= 
    \begin{cases} 
        \bm{{\on{parent}}}_{T^*}(c)\,
        &\mbox{if $\bm{{\on{child}}}_{T^*}(\bm{{\on{parent}}}_{T^*}(c),1) = c$} \\
        \texttt{null}, &\mbox{otherwise.} \\
    \end{cases} \vspace{2mm} \ \\ \ \\ 
    \bm{{\on{below}}}_P(c) &= 
    \begin{cases}
        \bm{{\on{child}}}_{T^*}(c,1),  &\mbox{if $\bm{{\on{parent}}}_{T^*}(c+1) = c$,} \\
        \texttt{null}, &\mbox{otherwise.}
    \end{cases}
\end{align*}

\item[]
{\textbf{Visibility Queries [$\bm{{\on{is\_visible}}}_{P}(c_1,c_2)$]:}}

To determine the visibility of two cells $c_1$ and $c_2$, where $c_1<c_2$, 
we first check whether $c_1$ and $c_2$ are on the same row by evaluating
$\bm{{\on{depth}}_{T^*}}(c_1) = \bm{{\on{depth}}_{T^*}}(c_2)$. If they are, 
then $c_1$ and $c_2$ are visible via a horizontal line if and only if all bits associated with cells after $c_1$ and up to $c_2$ in the left bitstring $L$ are 1, that is, 
$\bm{{\on{rank}}}_L(w_{c_2}) - \bm{{\on{rank}}}_L(w_{c_1}) = w_{c_2}-w_{c_1}$.
Suppose $c_1$ and $c_2$ are not on the same row. Then, by the construction of $T^*$, $c_1$ and $c_2$ are visible if and only if $c_2$ is a left-most descendent of $c_1$, that is, $c_2$ must find $c_1$ as an ancestor and all cells on the path from $c_1$ to $c_2$ appear consecutively in the pre-order traversal of $T^*$. In particular, $c_1$ and $c_2$ are visible if and only if 
$c_2 - c_1 = \bm{{\on{depth}}}_{T^*}(c_2) - \bm{{\on{depth}}}_{T^*}(c_1)$.
Refer to Algorithm~\ref{alg:is_vis} for the pseudocode.
\begin{center}
\scalebox{.9}{
\begin{minipage}{13cm}
    \begin{algorithm}[H]
\caption{Visibility Query for Nice Polyominoes}
\label{alg:is_vis}
\begin{algorithmic}
    \Function{$\bm{{\on{is\_visible}}}_P$}{$c_1, c_2$}
    \If{$depth_{T^*}(c_1) = depth_{T^*}(c_2)$}
        \State {$w_{c_1} \leftarrow \textbf{level\_order\_rank}_{T*}(c_1)$}
        \State {$w_{c_2} \leftarrow \textbf{level\_order\_rank}_{T*}(c_2)$}
        \If{$(rank_{L}(w_{c_2}) - rank_{L}(w_{c_1})) = (w_{c_2} - w_{c_1})$}
        \State \Return{\texttt{true}}  
        \EndIf
    \ElsIf{($c_2 - c_1) = (depth_{T^*}(c_2) - depth_{T^*}(c_1))$}
        \State \Return{\texttt{true}}
    \EndIf
    \State \Return{\texttt{false}}
    \EndFunction
\end{algorithmic}    
\end{algorithm}
\end{minipage}}
\end{center}

\end{itemize}

\end{proof}

\begin{figure}[t]
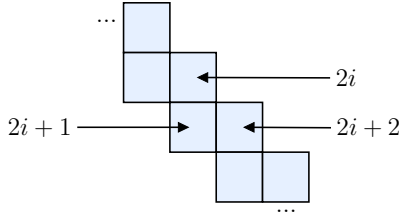

\centering
\scalebox{0.8}{\figStaircaseLabelledColoured}
    \caption{A staircase polyomino with $n$ cells formed such that any cell $c_{2i+1}$ finds $c_{2i}$ on its top and $c_{2i+2}$ on its right, for any $i\leq \lfloor n/2 \rfloor$ }
    \label{fig:staircase_polyomino}
\end{figure}

\subsection{Compact Data Structure for General Polyominoes}
The data structure of Theorem~\ref{th:mainBounded}
can be used to store any arbitrary polyomino compactly, 
but less space-efficiently than $3n+o(n)$. 
To see this, think of a \emph{staircase polyomino} formed by $n$ cells
(see Figure~\ref{fig:staircase_polyomino}). 
For this polyomino to fit in a strip, the height and width of the strip must be at least $n/2$.
Hence, counting dummy nodes, the left bitstring $L$ takes $n + n/2 + o(n)$ bits, and the covering tree $T^*$ takes $2(n + n/2) + o(n)$ bits, summing up to $4.5n + o(n)$ bits. In this section, we explain how this can be improved. 

Given any polyomino $P$ with $n$ cells and a function $f(n) \in \omega(1)$, 
we partition $P$ into \emph{slices} via horizontal cuts. Specifically, the first slice in the \emph{slicing of type} $i$, for any $i \in [1,f(n)]$, contains the first $i$ rows in $P$, and the remainder of $P$ is partitioned into further slices, each formed by the top $f(n)$ rows among the remaining rows. The last slice may have less than $f(n)$ rows (see Figure~\ref{fig:slicing}). 

For a given slicing
the \emph{top cells} are cells in the first row of each slice $s_j$ for $j \geq 2$, and \emph{bottom cells} are the cells that are in the last row of slice 1 or the row $f(n)$ of any slice $j\geq 2$.  
We refer to top and bottom cells as the \emph{boundary cells}.
Note that if a cell is a top (respectively bottom) cell for a given slicing of type $i$, it cannot be a top (respectively bottom) cell for any other slicing of type $i'\neq i$. Therefore, the number of boundary cells across all possible slicings is $2n$, implying that there is a slicing $i^*$---the one with a minimum number of boundary cells---for which the number of boundary cells is at most $2n/f(n)$. 

Let $h^*$ denote the number of slices in the slicing of type $i^*$. Arrange slices $s_1, \ldots, s_{h^*}$,  
from left to right, to form a (disconnected) polyomino $P_d$ such that the top vertices appear in the first row of $P_d$. Given that there are at most $f(n)$ rows in each slice, $P_d$ fits within a strip of height $f(n)$. We adapt the data structure of Theorem~\ref{th:mainBounded} to store $P_d$. In particular, we form a covering tree $T^*$ formed by $n+f(n)+1$ nodes
and a left bitstring $L$ of $n+f(n)+1$ bits. 
Recall that each cell is represented by its index in the pre-order traversal of $T^*$. Note that top cells appear at level $1$ of $T^*$, bottom cells of $s_1$ are at level $i^*$, and other bottom cells are at level $f(n)$.

We also store two bitstrings $\on{Top}$ and $\on{Bot}$ to encode connections between slices. Precisely, we visit top (respectively bottom) cells, in the level-order as they appear in $T^*$, and store a bit indicating whether the cell is connected to a cell above (respectively below). We store $\on{Top}$ and $\on{Bot}$ bitstrings using the encoding of Lemma~\ref{thm:succinct_bit_string}.  
%
%
Observe that the cell corresponding to the $i$'th $1$ in $\on{Top}$ is connected to the cell corresponding to the $i$'th $1$ in $\on{Bot}$. 
We also store two integers FirstTop and FirstBot, which respectively encode the number of cells in the topmost and bottommost rows of the first slice. 

\begin{figure}
\centering
\begin{subfigure}{\textwidth}
\centering
\scalebox{1.3}{\blueSlicinga}
\end{subfigure}\hskip1ex
\begin{subfigure}{\textwidth}
\centering
\scalebox{1.3}{\blueSlicingb}
\end{subfigure}
{\caption{ 
Slicing of type $i^*=2$ of a polyomino with $f(n) = 4$. Boundary cells are shaded.\label{fig:slicing}} }
\end{figure}

\begin{theorem}\label{th:mainGeneral}
For an arbitrary function 
$f(n) \in \omega(1)$, 
it is possible to store polyominoes of size $n$ in $3(n + f(n)) + o(n)$ bits, while supporting navigation queries in ${\mathcal{O}(1)}$ and visibility queries in $\mathcal{O}(n/f(n))$.
\end{theorem}

\begin{proof}
There are $f(n)+1$ dummy nodes in $P_d$. Therefore, by Lemmas~\ref{thm:succinct_bit_string} and \ref{thm:succinct_ordinal_tree}, we use $2(n+f(n))+o(n)$ bits for $T^*$, $n+f(n)+o(n)$ bits for $L$, $2\left(2n/f(n) + o(n/f(n))\right)$ bits for $\on{Top}$ and $\on{Bot}$, and $O(\log n)$ bits for FirstTop and FirstBot. 
Since $f(n) \in \omega(1)$, the total space is $3(n +f(n))+ o(n)$. Next, we describe how queries are answered. 
\begin{itemize}
\item[] 
\textbf{Navigation queries:}
Given a cell $c$ in the input polyomino $P$, the left and right neighbors of $c$ are the same as those in $P_d$, and we can report them in $\mathcal{O}(1)$ as described in Theorem~\ref{th:mainBounded}.
Next, we describe how to report the top and bottom neighbor of $c$, as illustrated in Algorithm~\ref{alg:slicing-above} and Algorithm~\ref{alg:slicing-below}, respectively.
\begin{enumerate}
    \item $\bm{{\on{above}}}_P(c)$: if $\bm{{\on{is\_top}}}_P = \texttt{false}$, $c$ is not a top cell and we compute $\bm{{\on{above}}}_P(c)$ as in Theorem~\ref{th:mainBounded}; otherwise, $c$ is a top cell.
    The rank of $c$ among all top cells can be computed as $k \leftarrow \bm{{\on{level\_rank}}}_{T^*}(c) - \text{FirstTop}-1$.
    If $\text{Top}[k] = 0$, then $c$ has no neighbor above it, and we return \texttt{null}. 
    Otherwise, $c$ is connected to the cell $c'$ corresponding to the $k$'th $1$ in $\on{Bot}$. 
    We compute $r_{c'} \leftarrow \bm{{\on{select}}}_{\on{Bot}}(k)$, giving us the index of $c'$ among all bottom nodes. If $r_{c'} < \text{FirstBot}$, $c'$ is in the first slice 
    (at depth $i^*$ of $T^*$), and we return $\bm{{\on{level\_select}}}_{T^*}(i^*,r_{c'}+1)$; 
    otherwise, $c'$ is at level $f(n)$ and we return $\bm{{\on{level\_select}}}_{T^*}(f(i),r_{c'}-\text{FirstBot}+1)$. 
    \item $\bm{{\on{below}}}_P(c)$: if $\bm{{\on{is\_bot}}}_P = \texttt{false}$, $c$ is not a bottom cell and we compute $\bm{{\on{bottom}}}_P(c)$ as in Theorem~\ref{th:mainBounded}; otherwise, $c$ is a bottom cell.
    We first find the rank $k$ of $c$ among all bottom cells.
    If $\bm{{\on{depth}}}_{T^*}(c) = i^*$, $c$ belongs to the first slice and we set  $k \leftarrow \bm{{\on{level\_rank}}}_{T^*}(c) -1 $, 
    and if $\bm{{\on{depth}}}_{T^*}(c) = f(n)$, $c$ belongs to another slice and we set $k \leftarrow \bm{{\on{level\_rank}}}_{T^*}(c) + \text{FirstBot} -1$. 
    If $\on{Bot}[k] = 0$, then $c$ has no neighbor below it, and we return \texttt{null}. 
    Otherwise, $c$ is connected to the cell $c'$ corresponding to the $k$'th $1$ in $\on{Top}$. 
    Let $r_{c'} \leftarrow \bm{{\on{select}}}_{\on{Top}}(k)$.
    Then $c'$ is at level 1 and finds other $\text{FirstTop} + r_{c'}+1$ nodes on its left at level 1, and we return $\bm{{\on{level\_select}}}_{T^*}(1,\text{FirstTop}+r_{c'}+1)$.
\end{enumerate}
\item []
\textbf{Visibility queries [$\bm{{\on{is\_visible}}}_{P}(c_1,c_2)$]:}
Suppose we are given two cells $c_1$ and $c_2$ in $P$ with $c_1 < c_2$. We describe how to report the visibility of $c_1$ and $c_2$, as described in  Algorithm~\ref{alg:is_vis_slicing}.

First, we use $T^*$ and $L$ to check whether $c_1$ and $c_2$ are visible horizontally, as described in Theorem~\ref{th:mainBounded}. In particular, we return \texttt{true} if $\bm{{\on{is\_visible}}}_{p_d}(c_1,c_2)$ returns \texttt{true}. 
Otherwise, we iteratively scan through slices to check if $c_1$ and $c_2$ are visible through a vertical line. Let $c'$ be a node that is initially $c_2$ and represents a node visible from $c_2$ at the slide considered in the current iteration. 
We repeatedly apply the following process: (i) If  
$c'-c_1 = \bm{{\on{depth}}}_{T^*}(c') - \bm{{\on{depth}}}_{T^*}(c_1)$, return true. 
This is similar to checking whether $c'$ and $c_1$ are visible in $P_d$, (see Theorem~\ref{th:mainBounded}). (ii) Otherwise, 
find the ancestor of $c'$ at level $1$ of $T^*$ via $anc = \bm{{\on{ancestor}}}_{T^*}(c',1)$. 
Now, if $c'$ is not a leftmost descendant of $anc$, then no cell on the first row of the current slice is visible to $c'$, and we return \texttt{false}. That is, if $c'-anc \neq  \bm{{\on{depth}}}_{T^*}(c') - \bm{{\on{depth}}}_{T^*}(anc)$, return \texttt{false}. (iii) Otherwise, update $c'$ to be $\bm{{\on{above}}}_{P}(anc)$. If $c'$ becomes \texttt{null}, return \texttt{false} (the ray moving upwards from $c_2$ hits a boundary of $P$ at $anc$ without passing through $c_1$); moreover, if $c'$ becomes smaller than $c_1$, return  \texttt{false} (the nodes visible from $c_2$ in the above slides can only appear before $c'$ in preorder). Otherwise, when $c'$ exists and is at least $c_1$, repeat by going to Step (i). 
In the worst case, we may scan over all slices before returning \texttt{false} at Step (iii). The running time is thus proportional to the number of slices, which is $\mathcal{O}(n/f(n))$.

The running time is proportional to the number of slices, which is $\mathcal{O}(n/f(n))$.
To illustrate, suppose $c_1 \leftarrow b$ and $c_2 \leftarrow a$ in Figure~\ref{fig:slicing}. Thus, $c'$ is initially $a$. At step (i), we check if $a$ is a leftmost descendent of $b$ in $T^*$. Since it is not, at step (ii), we find $anc$ as the ancestor of $a$ at level 1, which is $d$. Given that $a$ is a leftmost descendant of $d$ in $T^*$, we move on to step (iii), where we find the cell above $d$ in $P$, which exists and is $e$. Therefore, we update $c'$ to be $e$, which remains larger than $b$.  
Therefore, we start the new iteration by checking whether $e$ is a leftmost descendent of $b$ in step (i). Since it is not, at step (ii), we find $p$ as the ancestor of $e$ at level 1. Since $e$ is not the leftmost descendent of $p$, we return \texttt{false}. 
\end{itemize}
\end{proof}
\algoAboveSlicing
\algoBelowSlicing
\algoVisSlicing

\newpage
\vspace{2em}
\noindent In Theorem~\ref{th:mainGeneral}, setting $f(n) = \left\lceil \varepsilon \cdot \frac{n}{3} \right\rceil$ and $f(n) = \left\lceil \frac{n}{t(n)} \right\rceil$, for any $t(n) \in \omega(1)$, gives:

\begin{corollary}\label{corolOne}
It is possible to answer neighborhood queries in a polyomino $P$ with $n$ cells in $\mathcal{O}(1)$ using an oracle that takes $(3+\epsilon)n$ bits of space and answers visibility queries in $\mathcal{O}(1)$, for any constant $\epsilon > 0$. 
\end{corollary}

\begin{corollary}\label{corolTwo}
It is possible to answer neighborhood queries in a polyomino $P$ with $n$ cells in $\mathcal{O}(1)$ using an oracle that takes $3n+o(n)$ bits of space and answers visibility queries in $\mathcal{O}(t(n))$, for any function $t(n) \in \omega(1)$.
\end{corollary}
\vspace*{-2mm}

\section{Succinct Data Structure for Bar Graphs}
A \emph{bar graph}, also known as a \emph{wall polyomino}~\cite{BLPS96,Fer96}, is a polyomino $G$ consisting of $m$ \emph{bars} or \emph{columns} of cells 
such that the lowermost cell of bar $i$ shares an edge with the lowermost cell of bar $i+1$, for $i<m$, 
as shown in Figure~\ref{fig:basicBarGraph}.
For example, there are four distinct bar graphs with $n=3$ cells, namely, \scalebox{.35}{\BarrOne}, \scalebox{.35}{\BarrTwo}, \scalebox{.35}{\BarrThree}, and \scalebox{.35}{\BarrFour}, formed respectively by $3,2,2,$ and $1$ bars. 
We often discuss the \emph{canonical ordering} of cells in a bar graph, in which cells appear bar by bar from the lowest cell of each bar to its highest. For example, the canonical ordering of \scalebox{.5}{\BarrSeven} is $a,b,c,d,e,f,g$. 
Observe that there is a bijection between the collection of compositions of an integer $n$ (a representation of $n$ as a sum of positive integers) and the collection of bar graphs formed by $n$ cells. Since the number of distinct compositions of $n$ is $2^{n-1}$, the following holds.
\vspace*{-2mm}
\begin{proposition}
At least $n-1$ bits are needed to store a bar graph with $n$ cells.
\end{proposition}
\begin{proof}
Observe that there exists a bijection between the collection of compositions of $n$ and the collection of bar graphs of size $n$.
A composition of a positive integer $n$ is a way of writing $n$ as the sum of a sequence of positive integers. 
It is straightforward to count the number of possible compositions for an integer $n$ as follows.
For a positive integer $n$, we can have a string of $n$ 1's and place either a `+' or `,' sign in a spot between two adjacent 1's.
There are $n-1$ available spots between $n$ 1's, with a binary choice (`+' or `,') to fill each spot, so there are $2^{n-1}$ distinct compositions.
Each bar graph has a corresponding composition, and vice versa.
For example, when $n = 6$, one possible string is $1+1,1+1+1,1$ creating the composition sequence $\{2,3,1\} = 2+3+1$. 
The corresponding bar graph comprises bars of sizes 2, 3, and 1 placed side-by-side from left to right.
The number of distinct compositions of $n$ is $2^{n-1}$, and so, at least 
$n-1$ bits are needed to distinguish all bar graphs of size $n$. 
\end{proof}

In what follows, we present a succinct data structure that stores a bar graph with $n$ cells using $n+o(n)$ bits. Our structure has the following components: \vspace{1mm}
\newcommand{\looktable}{R_k}
\begin{itemize}
\item[]
\textbf{Lookup Table $\bm{{\looktable}}$:}
Given a bar graph $G$ with $n$ cells, let $k \leftarrow \lfloor \log \log n \rfloor$ and define a lookup table $\looktable$ that has an entry for each possible bar graph $H$ of size $k$. Each cell in $H$ is referred to by its index in the canonical ordering of $H$. We store the following for the entry of $H$ in $\looktable$:
\begin{itemize}
    \item for each pair $(c,c')$ of cells in $H$, a bit $\bm{{\on{vis}}}(c,c')$ which is 1 iff $c$ and $c'$ are visible and 0 otherwise. 
    \item the height of the shortest bar among all bars except for the last bar in $H$. This height is bounded above by $k$, and will be stored in $\lceil \log(k)\rceil$ bits.
\end{itemize}
Therefore, for each entry $H$, we store $\binom{k}{2} + \lceil \log(k) \rceil = \mathcal{O}(k^2)$ bits. Since $2^{k-1}$ bar graphs of size $k$ exist,
the size of $\looktable$ is 
$\mathcal{O}(2^{k} k^2) = o(n)$. 
\vspace{1.5mm}

\begin{figure}[t]
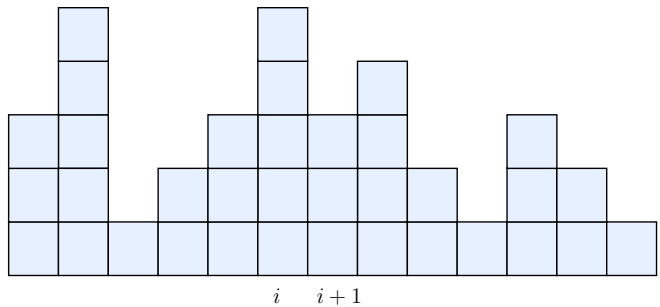

\centering
\scalebox{.7}{\generalBarGraphColoured}\vspace*{0mm}
\caption{A bar graph with $m = 13$ bars. The lowermost cell of bar $i$ is adjacent to the lowermost cell of bar $i+1$.}
\label{fig:basicBarGraph}
\end{figure}

\item[]
{\textbf{Block Partitioning:}}
We partition $G$ into smaller bar graphs, referred to as \emph{blocks} $B_1, B_2, \ldots, B_\ell$ as follows. To form the first block $B_1$, we process bars from left to right until the total number of cells in $B_1$ becomes at least $k$. 
At this point, $B_1$ is complete, and we recursively partition the remainder of $G$ into further blocks. 
The first $k$ cells in each block are called \emph{main cells}, and the remaining cells, if any, are called \emph{leftover cells}. In each block, all leftover cells are contained in the last bar.\vspace{2mm}

\item[]
{\textbf{Cartesian tree $\bm{{C_G}}$:}}
Having $G$  split into blocks, 
we determine the shortest bar in each block and store the number of cells in such bar in a list $s_1,s_2,\ldots,s_\ell$. Then, we create a \emph{Cartesian tree} $C_G$ based on this list: the root of $C_G$ is the minimum number among all values of $s_i$, and we recursively construct its left and right subtrees from the subsequences before and after this number.
Since the number of blocks is $\ell \leq n/\log \log n = o(n)$, 
we can store $C_G$ using $o(n)$ bits. Note that we do not store the values $s_i$, only the structure of $C_G$. 
\vspace{2mm}

\item[]
{\textbf{Bitstrings $\bm{{S_G}}$ and $\bm{{B_G}}$:}}
Consider a bitstring $S_G$ that has a bit for each cell in $G$ in the canonical order. This 
bit is $1$ if the cell is the lowest in its bar and $0$ otherwise. 
Using Lemma~\ref{thm:succinct_bit_string}, we store $S_G$ in $n+o(n)$ bits and refer to each cell by its index in $S_G$. 
We store a second bitstring $B_G$, which also has a bit for each cell $c$ in the canonical order, where the bit of $c$ is 1 if $c$ the first cell in its block and 0 otherwise. 
The number of $1$'s in $B_G$ is at most $n/ \log\log n$.
Using Lemma~\ref{thm:succinct_bit_string_k}, where $k = n/ \log\log n$,
the data structure takes $o(n)$ bits of space to store $B_G$ as illustrated in what follows:

\begin{align*}
\log \binom{n}{\frac{n}{\log\log n}} + o(n)   
& = \frac{n}{\log\log n}  \log(\frac{n}{\frac{n}{\log\log n}}) + o(n) \\
& = \frac{n}{\log\log n}  \log\log\log n + o(n)
\in o(n)  
\end{align*}

 \end{itemize}
%
\begin{theorem}
\label{th:barGraph}
There exists a data structure which stores bar graphs of size $n$ using $n + o(n)$ bits of space while supporting neighborhood and visibility queries in $\mathcal{O}(1)$.
\end{theorem}

\begin{proof}
In the data structure described above, 
$\looktable$ takes $o(n)$ bits, 
$C_G$ takes $o(n)$  bits, 
$S_G$ takes $n + o(n)$ bits, and 
$B_G$ takes $o(n)$, summing to $n+o(n)$ bits in total. 

Before describing how the queries are supported,
we review a few useful operations that can be supported in $\mathcal{O}(1)$ using Lemma~\ref{thm:succinct_bit_string}. 
Recall that a cell $x$ is referred to by its index in $S_G$. 
\begin{center}
\scalebox{.99}{
    \begin{tabular}{|l|l|}
    \hline 
        $\bm{{\on{bar}}}_G(x)$ & rank of the bar that $x$ belongs to, among all bars \\ 
        \hline  
        $\bm{{\on{first\_in\_bar}}}_G(x)$ & the bottommost cell in the bar that contains $x$  \\ 
        \hline 
        $\bm{{\on{height}}}_G(x)$ & the number of cells below $x$ in its bar \\ 
        \hline 
        $\bm{{\on{block}}}_G(x)$ & the index of the block that $x$ belongs, among all blocks \\ 
        \hline 
        $\bm{{\on{size}}}_G(i)$  & number of cells in bar of index $i$ \\ 
        \hline
    \end{tabular}
}
\end{center}
We can find the rank of the bar that $x$ belongs to, among all bars, as $\bm{{\on{bar}}}_G(x) = \bm{{\on{rank}}}_{S_G}(x)$. 
The bottommost cell in the bar that contains $x$ is then $\bm{{\on{first\_in\_bar}}}_G(x) = \bm{{\on{select}}}_{S_G}(\bm{{\on{bar}}}(x))$, and the \emph{height} of $x$, defined as the number of cells below $x$ in its bar, is $\bm{{\on{height}}}_G(x) = x - \bm{{\on{first\_in\_bar}}}_G(x)$.
The index of the block that $x$ belongs, among all blocks, is $\bm{{\on{block}}}_G(x) = \bm{{\on{rank}}}_{B_G}(x)$. 
The number of cells in a bar of index $i$ among all bars is  $\bm{{\on{size}}}_G(i) = \bm{{\on{select}}}_{S_G}(i+1) - \bm{{\on{select}}}_{S_G}(i)$, and if $\bm{{\on{select}}}_{S_G}(i+1)$ returns \texttt{null}, then $i$ is the last column, and $\bm{{\on{size}}}_G(i)  = n-\bm{{\on{select}}}_{S_G}(i)$.

\begin{itemize}
\item[] 
{\textbf{Neighborhood queries:}}
Neighborhood queries are answered only using $S_G$. 
Given a cell $x$, we describe how to answer neighborhood queries.
\begin{itemize}
    \item $\bm{{\on{above}}}_G(x)$: if $S_G[x+1] = 1$, then $x$ is the uppermost cell in its bar and return \texttt{null}; otherwise, return $x+1$. 
    \item $\bm{{\on{below}}}_G(x)$: if $S_G[x] =1$, then $x$ is the lowermost cell and there is no cell below, return \texttt{null}; otherwise, return $x-1$.
    \item $\bm{{\on{left}}}_G(x)$: if $\bm{{\on{bar}}}_G(x) = 0$, we return \texttt{null} since $x$ appears in the leftmost bar; otherwise, we compute the bottommost cell $c_\ell$ in the bar immediately on the left of $x$ as 
    $c_\ell \leftarrow \bm{{\on{select}}}_{S_G}(\bm{{\on{bar}}}_G(x)-1)$. 
    Then if $c_\ell + \bm{{\on{height}}}_G(x) < \bm{{\on{first\_in\_bar}}}_G(x)$, there is a cell of the same height as $x$ on the left of $x$, and we return $c_\ell+\bm{{\on{height}}}_G(x)$; otherwise, return \texttt{null}.
    \item $\bm{{\on{right}}}_G(x)$: set $c_r \leftarrow \bm{{\on{select}}}_{S_G}(\bm{{\on{bar}}}_G(x)+1)$ as the bottommost cell in the bar on the right of $x$. 
    If $c_r$ is \texttt{null}, then $x$ is in the rightmost bar, and we return $\texttt{null}$. 
    Otherwise, if $c_r + \bm{{\on{height}}}_G(x) \leq \bm{{\on{size}}}_G(\bm{{\on{bar}}}_G(x)+1)$, there is a cell of the same height as $x$ on the right of $x$, and we return  $c_r + \bm{{\on{height}}}_G(x)$; otherwise return \texttt{null}. 
\end{itemize}

\algoNeighbourhoodBar

\begin{figure}[t]
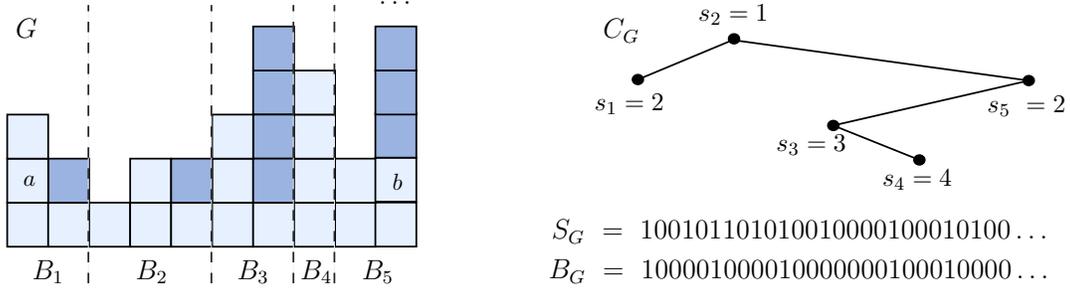

\centering
\scalebox{.88}{\barGraphLabelledColoured}
\caption{An example of a bar graph $G$ with $n = 29 + 2^{16}$ cells split into blocks, with $k = \lfloor \log \log n  \rfloor=4$. 
Only the first 29 cells in the canonical order are shown; the remaining $2^{16}$ belong to the last bar. 
Leftover cells are darker. }
\label{fig:bar_graph_full_picture}
\end{figure}
%
%
%
\vspace{2em}
\item[]
{\textbf{Visibility Queries:}}
%
Given two cells $x_1$ and $x_2$ from a bar graph $G$, where $x_1 < x_2$, we apply the following procedure to answer visibility queries, as illustrated in Algorithm~\ref{alg:is_vis_bar}: 
\begin{enumerate}
    \itemsep0em
    \item (Check if $x_1$ and $x_2$ belong to same bar.) If $\bm{{\on{bar}}}_G(x_1) = \bm{{\on{bar}}}_G(x_2)$, $x_1$ and $x_2$ belong to the same bar and is thus visible; return \texttt{true}.
    \item (Check if $x_1$ and $x_2$ have different heights.) If $\bm{{\on{height}}}_G(x_1) \neq \bm{{\on{height}}}_G(x_2)$, $x_1$ and $x_2$ have different heights and bars, and is not visible; return \texttt{false}.
    \item (Check if $x_1$ and $x_2$ are bottommost cells.) If $S_G[x_1] = S_G[x_2] = 1$, $x_1$ and $x_2$ are the bottommost cells in the bar; return \texttt{true}. 
    \item (The non-trivial case.) 
    In what follows, $x_1$ and $x_2$ belong to different bars. 
    Set $i \leftarrow \bm{{\on{block}}}_G(x_1)$ and $j \leftarrow \bm{{\on{block}}}_G(x_2)$.
    \vspace{-0.5em}
    \begin{enumerate}
        \item ($x_1$ and $x_2$ belong to the same block.) If $i=j$, $x_1$ and $x_2$ belong to the same block and we perform the following steps:
        \begin{enumerate}
        \itemsep0em
        \item Set $q \leftarrow \bm{{\on{select}}}_{B_G}(i)$.
        This is the first cell in $B_i$.
        \item Set $entry \leftarrow S_G[q, q + k]$. 
        These bits specify the entry associated with the bar graph $H$ of size $k$ in the lookup table $R_k$. Note that we can read the entry in $\mathcal{O}(1)$ under the word RAM model assumption.
        \item Evaluate if $x_2$ is a main cell via $q+k < x_2$.
        If \texttt{true}, set $y \leftarrow x_2$.
        Otherwise, when $x_2$ is a leftover cell, set $y = \bm{{\on{left}}}_G(x_2)$.
        If $y = \texttt{null}$, there is no cell on the left of $x_2$ so it cannot see $x_1$; return \texttt{null}.
        \item Now $x_1$ and $x_2$ are visible if $x_1$ and $y$ are visible. Note that $x_1$ and $y$ are the main vertices with indices $x_1-q$ and $y-q$ in $H$.
        We evaluate and return $S_G[entry].\bm{{\on{vis}}}(x_1-q,y-q) = 1$.
    \end{enumerate}
    
    \item ($x_1$ and $x_2$ belong to different blocks.) When $i \neq j$, $x_1$ and $x_2$ belong to different blocks and we perform the following steps:
    \begin{enumerate}
        \itemsep0em
        \item Set $h \leftarrow \bm{{\on{height}}}_G(x_1)$; that is, $h$ is the height of $x_1$ and $x_2$. 
        \item (Find the rightmost cell $w$ at height $h$ in $B_i$.) 
        Let $c_{last_i}$ be the last cell in canonical ordering that belongs to $B_i$. 
        Set $c_{last_i} \leftarrow \bm{{\on{select}}}_{B_G}(i+1)-1$.
        Next, we verify that the last bar of $B_i$ has height at least $h$. 
        That is, if $\bm{{\on{height}}}_G(c_{last_i}) \geq h$, set $w \leftarrow \bm{{\on{first\_in\_bar}}}_G(c_{last_i})+h - 1$.
        Otherwise, set $w \leftarrow \texttt{null}$.
        \item (Find the leftmost cell $z$ at height $h$ in $B_j$.) 
        Finding $z$ is done symmetrically to $w$.
        Let $c_{first_j}$ be the first cell in canonical ordering that belongs to $B_j$.
        Set $c_{first_j} \leftarrow \bm{{\on{select}}}_{B_G}(j)$.
        Again, we verify that the first bar of $B_j$ has height at least $h$.
        If $\bm{{\on{height}}}_G(c_{first_j}) \geq h$, set $z \leftarrow \bm{{\on{first\_in\_bar}}}_G(c_{first_j}) + h - 1$.
        Otherwise, set $z \leftarrow \texttt{null}$
        \item (Verify existence of $w$ and $z$.)
        If $w=\texttt{null}$ or $z=\texttt{null}$, there exists a bar between $x_1$ and $x_2$ with height less than $h$; return \texttt{false}. 
        \item (Check visibility within blocks $B_i$ and $B_j$.) 
        Given that $x_1$ and $w$ (respectively $x_2$ and $z$) belong to the same block, we can check their visibility by referring to the lookup table, as described in Step 4a. If either of these pairs are not visible, output \texttt{false}. 
        \item (Check visibility between blocks $B_i$ and $B_j$.)
        Let $\ell$ be the least common ancestor of nodes $i$ and $j$ in the Cartesian tree $C_G$.
        Set $\ell \leftarrow \bm{{\on{LCA}}}_{C_G}(i, j)$.
        $B_\ell$ will be the block in between $B_i$ and $B_j$ whose minimum height is the smallest among all blocks between $B_i$ and $B_j$. 
        Therefore, $w$ and $z$ are visible iff $h \geq lo$, where $lo$ is the smallest height of any bar in $B_\ell$. 
        To find $lo$, we read the entry associated with $B_\ell$ in the lookup table $R_k$ to determine the smallest bar height for all bars in $B_\ell$ except the last bar. 
        Set $h_1 \leftarrow R_k[B_\ell].height$.  
        We further find the height $h_2$ of the last bar in $B_\ell$, which includes the leftover vertices. 
        For that, find $c_{last_\ell} \leftarrow \bm{{\on{select}}}_{B_G}(\ell+1)-1$ as the last cell in $B_\ell$, and let $h_2 \leftarrow \bm{{\on{height}}}(c_{last_\ell})$.
        Set $lo \leftarrow \on{min}\{h_1, h_2\}$.
        If $lo \geq h$, return \texttt{true}. Otherwise, return \texttt{false}.
    \end{enumerate}
\end{enumerate}
\end{enumerate}
\end{itemize}

\end{proof}

\noindent\textbf{Example:} 
Assume we want to report the visibility of cells $a$ and $b$ in Figure~\ref{fig:bar_graph_full_picture}. 
We follow the steps in Algorithm~\ref{alg:is_vis_bar}.
Recall that cells in a bar graph $G$ are referred to by its index in $S_G$.
We first check if $a$ and $b$ belong to the same bar by evaluating $\bm{{\on{bar}}}_G(a) = \bm{{\on{bar}}}_G(b) \iff 1 \neq 10$, so cells $a$ and $b$ are not visible along a bar.
Next, we check if cells $a$ and $b$ have the same height, $\bm{{\on{height}}}_G(a) = \bm{{\on{height}}}_G(b) \iff 2 = 2$. 
Set $h \leftarrow \bm{{\on{height}}}_G(a)$.
Since the equality holds, we check that the bars between $a$ and $b$ have a height of at least $h$. 
%
$a$ belongs to block $i \leftarrow 1$ and $b$ belongs to block $j \leftarrow 5$.
We select the last cell in the canonical ordering that belongs to $B_i$, 
$c_{last_i} \leftarrow \bm{{\on{select}}}_{B_G}(i+1)-1 = \bm{{\on{select}}}_{B_G}(2)-1 = 5-1 = 4$. 
The height of the last cell is at least $h$, $\bm{{\on{height}}}_G(c_{last_i}) \geq h \iff 2 \geq 2$.
We find the rightmost cell at height $h$ in block $i$ denoted as $w 
\leftarrow \bm{{\on{first\_in\_bar}}}_G(c_{last_i})+h-1 
= 3+2-1 = 4$.
Similarly, we select the first cell in canonical ordering that belongs to $B_j$,
$c_{first_j} \leftarrow \bm{{\on{select}}}_{B_G}(j) = 22$.
The height of the first cell in $B_j$ is at least $h$, $\bm{{\on{height}}}_G(c_{first_j}) \geq h \iff 2 \geq 2$
We find the leftmost cell at height $h$ in block $j$ denoted as $z
\leftarrow \bm{{\on{first\_in\_bar}}}_G(c_{first_j})+h-1 
= 22+2-1 = 23$.
Next, we check that $(a, w)$ ($(z, b)$ respectively) are visible by checking the lookup table $R_k$. Recall that $k \leftarrow \lfloor \log\log n \rfloor$. 
$\bm{{\on{vis}}}(a, w)$ ($\bm{{\on{vis}}}(z, b)$ respectively) both return 1, meaning the pair is visible within their block.
We check that $w$ and $z$ are visible by finding the LCA in $C_G$, $\ell \leftarrow 2$.
From entry $B_2$ in $R_k$, the smallest height of a bar in block $B_2$ is $h_1 \leftarrow 1$.
The last cell in $B_2$ is $c_{last_\ell} \leftarrow \bm{{\on{select}}}_{b_G}(\ell+1)-1 = 10 -1 = 9$.
The height of the last bar in $B_2$ is denoted as $h_2 \leftarrow \bm{{\on{height}}}(c_{last_\ell}) = 2$.
Let $lo \leftarrow \text{min}\{h_1, h_2\} = \text{min}\{1, 2\} = 1$.
The smallest bar in $B_\ell$ is smaller than the height of $a$ and $b$,
$lo < h \iff 1 < 2$; thus $a$ and $b$ are not visible, we return \texttt{false}.

\hspace*{-1cm}

\section{Concluding Remarks}
In this paper, we presented a compact data structure for polyominoes and a succinct data structure for bar graphs, both of which support navigation, adjacency and visibility queries. From this work, a number of unanswered open questions remain. 

We define \emph{distance queries} to be the queries that report the cells that are a distance $d$ to either direction (up, down, left, and right) from a given cell. Can these be supported in constant time while storing polyominoes compactly? 

We say a polyomino $P$ is \emph{column-convex} (respectively \emph{row-convex}) if 
all cells within the same column (respectively rows) are visible to each other.
A polyomino $P$ is \emph{convex} if it is both column-convex and row-convex (see Figure~\ref{fig:convexPp}).
Does there exist a succinct data structure for column-convex (or row-convex) polyominoes? 
Observe that bar graphs are column-convex polyominoes. 

\convexpoly

\algoVisBar


\newpage
\bibliographystyle{IEEEbib}
\bibliography{refs.bib}


\end{document}